\documentclass{cccg13}

\usepackage{graphicx,amssymb,amsmath,dsfont}

\usepackage{url}

\usepackage{verbatim}

\newcommand{\IR}{\mathbb{R}}
\newcommand{\C}{\mathcal{C}}

\newcommand{\T}{\mathcal{T}}
\newcommand{\F}{\mathcal{F}}
\newcommand{\CH}{Conv}

\title{Computing Covers of Plane Forests\thanks{The authors were supported by NSERC. 
A.B.\ was supported by Carleton University's I-CUREUS program.}}

\author{
Luis Barba\thanks{School of Computer Science, Carleton University, 
      Ottawa, Canada.} \thanks{Boursier FRIA du FNRS, D\'epartement d'Informatique, Universit\'e Libre de Bruxelles}
\and
Alexis Beingessner\footnotemark[2]
\and 
Prosenjit Bose\footnotemark[2]
\and 
Michiel Smid\footnotemark[2]
} 

\begin{document}
\thispagestyle{empty}
\maketitle


\begin{abstract}
Let $\phi$ be a function that maps any non-empty subset $A$ of 
$\IR^2$ to a non-empty subset $\phi(A)$ of $\IR^2$. A $\phi$-cover of 
a set $T=\{T_1, T_2, \dots, T_m\}$ of pairwise non-crossing trees 
in the plane is a set of pairwise disjoint connected regions such that
\begin{enumerate} 
\item each tree $T_i$ is contained in some region of the cover,
\item each region of the cover is either 
  \begin{enumerate}
    \item $\phi(T_i)$ for some $i$, or \label{prop:opt1}
    \item $\phi(A \cup B)$, where $A$ and $B$ are constructed by 
          either~\ref{prop:opt1} or~\ref{prop:opt2}, and 
          $A \cap B \neq \emptyset$. \label{prop:opt2}   
  \end{enumerate}
\end{enumerate}
We present two properties for the function $\phi$ that make the 
$\phi$-cover well-defined. Examples for such functions $\phi$ are 
the convex hull and the axis-aligned bounding box. For both of these 
functions $\phi$, we show that the $\phi$-cover can be computed in 
$O(n\log^2n)$ time, where $n$ is the total number of vertices of the 
trees in $T$. 
\end{abstract}


\section{Introduction}   \label{secintro} 
Let a \emph{geometric tree} be a plane straight-line embedding 
of a tree in $\IR^2$. Consider a set $T = \{T_1,T_2,\ldots,T_m\}$ 
of $m$ pairwise non-crossing geometric trees with a total of $n$ 
vertices in general position. 
The \emph{coverage} of these trees is the set of all points $p$ in 
$\IR^2$ such that every line through $p$ intersects at least one of 
the trees. Beingessner and Smid~\cite{Beingessner12} showed that the coverage 
can be computed in $O(m^2n^2)$ time. They also presented an example of 
$m=n/2$ pairwise non-crossing geometric trees (each one being a line segment) 
whose coverage has size $\Omega(n^4)$. Thus, the worst-case complexity of 
computing the coverage is $\Theta(n^4)$. 

Since the worst-case inputs are rather artificial, we consider the 
following heuristic for reducing the running time. Let $\CH$ denote 
the convex hull. We observe that the coverage of the trees in $T$ 
is equal to the coverage of their convex hulls. 
Moreover, if two convex hulls $\CH(T_i)$ and $\CH(T_j)$ overlap, then 
we can replace them by the convex hull of their union without changing 
the coverage. By repeating this process, we obtain a collection of 
pairwise disjoint convex polygons whose coverage is equal to the 
coverage of the input trees. Ideally, the number of these convex 
polygons and their total number of vertices are much less than $m$ 
and $n$, respectively. If this is the case, then running the 
algorithm of~\cite{Beingessner12} on the convex polygons gives the 
coverage of the input trees in a time that is much less than 
$\Theta(n^4)$ time, provided that we are able to quickly compute 
the collection of pairwise disjoint convex polygons. In this paper, 
we show that this is possible, by providing an 
$O(n \log^2 n)$--time algorithm. 

We now formally state the above process. 
\begin{enumerate} 
\item Let $\C = \{ \CH(T_i) \mid 1 \leq i \leq m \}$. 
\item While the elements of $\C$ are not pairwise disjoint: 
      \begin{enumerate} 
      \item Take two arbitrary elements $C$ and $C'$ in $\C$ for 
            which $C \neq C'$ and $C \cap C' \neq \emptyset$. 
      \item Let $C'' = \CH(C \cup C')$. 
      \item Set $\C = ( \C \setminus \{C,C'\} ) \cup \{ C''\}$.  
      \end{enumerate} 
\item Return the set $\C$. 
\end{enumerate} 

The output $\C$ is a collection of pairwise disjoint convex polygons, 
which we refer to as the \emph{hull-cover} of $T$. See Figure~\ref{hull-evaluation} for two examples. Since in Step 2(b), the two elements 
$C$ and $C'$ are chosen \emph{arbitrarily} (as long as they are 
distinct and overlap), the reader may object to the use of the 
word ``the'' in front of ``hull-cover''. In Section~\ref{sec:phi-cover} we
justify the use of this word by proving that, no matter 
what choices are made in Step 2(b), the output $\C$ is always the 
same. 

\begin{figure}[h!]
  \centering
  \includegraphics[scale=0.4]{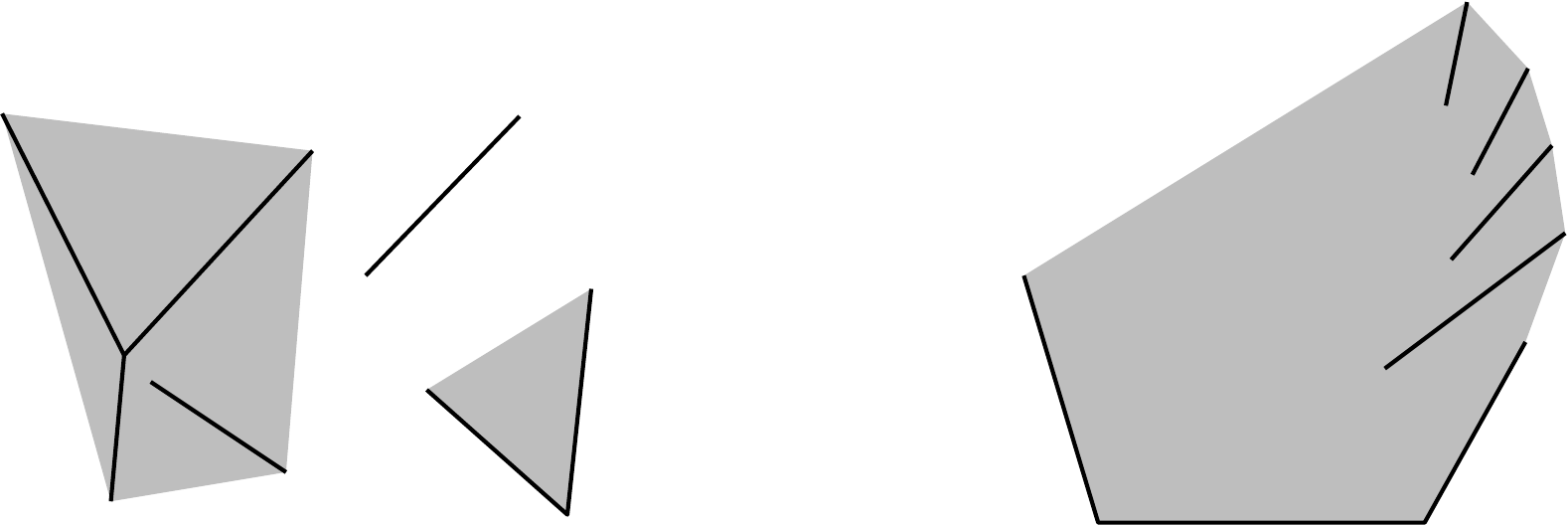}
  \caption{Two examples of hull-covers. Note that the hull-cover on the right demonstrates what is in some sense the worst case for the number of times intersections will need to be re-evaluated.}
  \label{hull-evaluation}
\end{figure} 

\section{$\phi$-Covers} \label{sec:phi-cover}
Consider a function $\phi$ that maps any non-empty subset $A$ of $\IR^2$ 
to a non-empty subset $\phi(A)$ of $\IR^2$. We assume that this function 
satisfies the following properties: 

\vspace{0.5em} 

\noindent 
{\bf Property 1:} For any non-empty subset $A$ of $\IR^2$, 
\[ A \subseteq \phi(A) .
\] 

\noindent 
{\bf Property 2:} For any two non-empty subsets $A$ and $B$ of $\IR^2$, 
\[ \mbox{if $A \subseteq \phi(B)$, then $\phi(A) \subseteq \phi(B)$.}
\]

Both the convex hull and axis-aligned bounding box functions satisfy 
these properties. However, the minimum enclosing circle function does 
not satisfy Property~2. 

We rewrite the algorithm described in Section~\ref{secintro} 
using the function $\phi$ instead of $\CH$. We also use a forest $\F$ 
of binary trees to keep track of the history of the process; each 
node $u$ in this forest stores a value $\phi(u)$. 
The forest helps us to prove that the $\phi$-cover is well-defined. 

\begin{enumerate} 
\item For each $i$ with $1 \leq i \leq m$, let $\T_i$ be the tree 
      consisting of the single node $r_i$, whose value $\phi(r_i)$ is 
      equal to $\phi(T_i)$. 
\item Initialize the forest $\F = \{ \T_i \mid 1 \leq i \leq m \}$. 
\item Let $\C = \{ \phi(r_i) \mid 1 \leq i \leq m \}$. 
\item While the elements of $\C$ are not pairwise disjoint: 
      \begin{enumerate} 
      \item Take two arbitrary roots $r$ and $r'$ in the forest $\F$ 
            for which $r \neq r'$ and $\phi(r) \cap \phi(r') \neq \emptyset$. 
      \item Let $\T$ and $\T'$ be the trees in $\F$ whose roots are 
            $r$ and $r'$, respectively. 
      \item Let $r''$ be a new node and set its value $\phi(r'')$ to 
            $\phi( \phi(r) \cup \phi(r') )$. 
      \item Create a new tree $\T''$ whose root is $r''$ and make 
            $\T$ and $\T'$ the two children of $r''$. 
      \item Set $\F = ( \F \setminus \{\T,\T'\} ) \cup \{ \T''\}$.  
      \item Set $\C = ( \C \setminus \{\phi(r),\phi(r')\} ) \cup \{ \phi(r'') \}$.  
      \end{enumerate} 
\item Return the forest $\F$ and the set $\C$. 
\end{enumerate} 
We refer to the output set $\C$ as the $\phi$-\emph{cover} 
of $T$. In Theorem~\ref{theorem0} below, we prove that the $\phi$-cover 
is well-defined. Before we prove this theorem, we present a third 
property of the function $\phi$:  

\vspace{0.5em} 

\noindent 
{\bf Property 3:} For any two non-empty subsets $A$ and $B$ of $\IR^2$, 
\[ \phi(A) \subseteq \phi(\phi(A) \cup \phi(B)) .
\]
Note that this property follows trivially from Property~1, because 
\[ \phi(A) \subseteq \phi(A) \cup \phi(B) \subseteq \phi(\phi(A) \cup \phi(B)) .
\]

\begin{lemma}  \label{lem:always-there}
Let $\C$ and $\C'$ be two $\phi$-covers with corresponding forests 
$\F$ and $\F'$, respectively. For each node $u$ in $\F$, there exists a 
root $r'$ in $\F'$ such that $\phi(u) \subseteq \phi(r')$. 
\end{lemma} 
\begin{proof} 
We prove the lemma by induction on the height of the subtree rooted 
at $u$. First assume that $u$ is a leaf in $\F$. Let $i$ be the index 
such that $\phi(u)=\phi(T_i)$, let $u'$ be the leaf in $\F'$ for which 
$\phi(u')=\phi(T_i)$, let $\T'$ be the tree in $\F'$ that has $u'$ as a leaf, 
and let $r'$ be the root of $\T'$. We prove that 
$\phi(u) \subseteq \phi(r')$.  

Let $u'_1=u', u'_2,\ldots,u'_k=r'$ be the nodes in $\T'$ on the path 
from $u'$ to $r'$. For each $i$ with $1 \leq i < k$, let $v'_i$ be the 
sibling of $u'_i$. Since 
\[ \phi(u'_{i+1}) = \phi( \phi(u'_i) \cup \phi(v'_i) ) ,
\] 
it follows from Property~3 that $\phi(u'_i) \subseteq \phi(u'_{i+1})$. 
From this, it follows that 
\[ \phi(u) = \phi(u') = \phi(u'_1) \subseteq \phi(u'_2) \subseteq \ldots 
               \subseteq \phi(u'_k) = \phi(r') . 
\]  

Now assume that $u$ is not a leaf. Let $v$ and $w$ be the children 
of $u$. Observe that $\phi(v) \cap \phi(w) \neq \emptyset$. 
By induction, there exist roots $r'$ and $r''$ in $\F'$ such that 
$\phi(v) \subseteq \phi(r')$ and $\phi(w) \subseteq \phi(r'')$. Since 
$\phi(r') \cap \phi(r'') \neq \emptyset$, we must have $r'=r''$. 
Thus, since $\phi(v) \cup \phi(w) \subseteq \phi(r')$, Property~2 implies 
that 
\[ \phi(u) = \phi( \phi(v) \cup \phi(w) ) \subseteq \phi(r') .
\] 
\end{proof} 

\begin{theorem} 
\label{theorem0}
The $\phi$-cover is well-defined. 
\end{theorem} 
\begin{proof} 
Let $\C$ and $\C'$ be two $\phi$-covers with corresponding forests 
$\F$ and $\F'$, respectively. We have to prove that $\C = \C'$. 
Observe that 
\[ \C = \{ \phi(r) \mid \mbox{ $r$ is a root in $\F$}\} 
\]
and 
\[ \C' = \{ \phi(r') \mid \mbox{ $r'$ is a root in $\F'$}\}. 
\]
Let $r$ be a root in $\F$. By Lemma~\ref{lem:always-there}, there exists a root 
$r'$ in $\F'$ such that $\phi(r) \subseteq \phi(r')$. Again by 
Lemma~\ref{lem:always-there}, applied with the roles of $\F$ and $\F'$ 
interchanged, there exists a root $r''$ in $\F$ such that 
$\phi(r') \subseteq \phi(r'')$. Thus, we have 
\[ \phi(r) \subseteq \phi(r') \subseteq \phi(r'') . 
\] 
Since $\phi(r) \cap \phi(r'') \neq \emptyset$, we must have $r=r''$. 
Therefore, $\phi(r)=\phi(r')$. We conclude that 
$\C \subseteq \C'$. By a symmetric argument, we can show that 
$\C' \subseteq \C$. 
\end{proof} 

Thus, the $\phi$-cover is well-defined for both the convex hull and the  
axis-aligned bounding box. If $\phi$ is the minimum enclosing circle 
function, then, in addition to not satisfying Property~2, the 
$\phi$-cover is not well-defined: In Figure~\ref{fig:min-circle}, an 
example is given for which the order in which merges are performed can 
result in different outputs. 

\begin{figure}[h]
  \centering
  \includegraphics[scale=1]{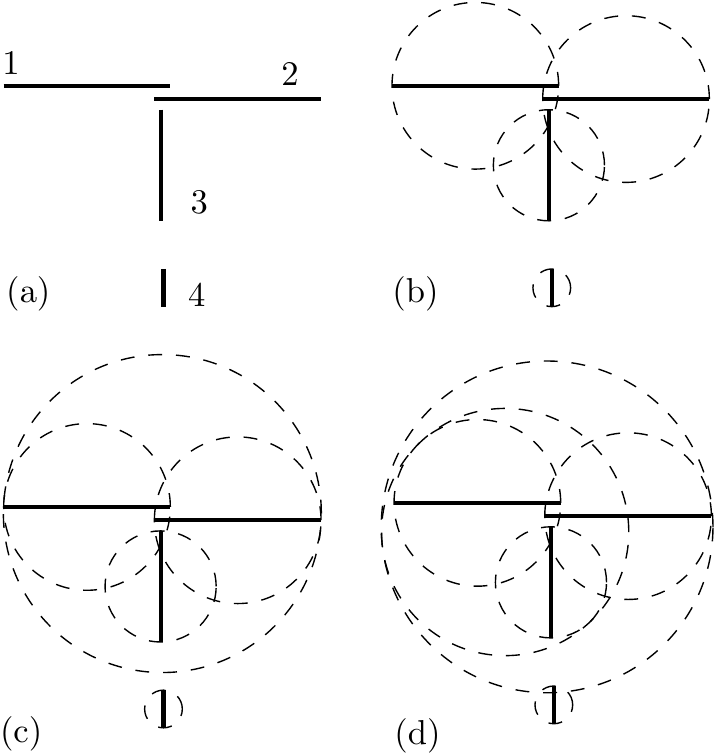}
  \caption{(a)~The input forest with trees numbered; (b)~The minimum enclosing circle of each tree; (c)~Merging 1 and 2 first results in no intersection with 4; (d)~Merging 1 and 3 first results in an intersection with 4.}
  \label{fig:min-circle}
\end{figure}


\section{Computing the Hull-Cover} 
In this section, we take for $\phi$ the convex hull function 
and show that the $\phi$-cover can be computed in $O(n \log^2 n)$ time. 

\subsection{Weakly Disjoint Polygons}
Finding the convex hull of two convex polygons can be a relatively expensive operation due to the fact that their boundaries can cross in $\Omega(n)$ different places. For example, consider a regular $n$-gon being merged with a copy of itself rotated $\epsilon$ degrees. In this section we demonstrate that, because our convex polygons are the convex hulls of disjoint trees, they behave much nicer than general convex polygons.

Let a \emph{weakly disjoint pair} of convex polygons $P$, $Q$ be a pair of convex polygons such that $P \setminus Q$ and $Q \setminus P$ are both connected sets of points, and $P$ does not share a vertex with $Q$. Then a \emph{weakly disjoint set} of polygons is a set of polygons such that all pairs of polygons are weakly disjoint. For simplicity, we assume that the convex hull of a line segment is a valid degenerate convex polygon consisting of two edges. We also assume all vertices are in general position. In this section we prove that weakly disjoint polygons are better behaved than general convex polygons, and that the convex hulls of disjoint trees are weakly disjoint. 

\begin{lemma}
\label{lem:pseudo-disks}
If two convex polygons $P, Q$ are weakly disjoint, then their boundaries intersect at at most two points.
\end{lemma} 

\begin{proof}
Assume the intersection of their boundaries, $\partial P \cap \partial Q$, contains more than two points. Further, assume without loss of generality that $P$ contains points outside of $Q$. Start at a point on $P$'s boundary $\partial P$ that is outside of $Q$, and walk along $\partial P$. Eventually we intersect $\partial Q$, and now $P$ is separated into two connected regions: points inside of $Q$, and points outside of $Q$. If we continue walking along $\partial P$, we eventually cross $\partial Q$ again. Now there are three regions of $P$: two outside $Q$, and one inside $Q$, but the two outside $Q$ may be the same. Continuing along $\partial P$ we must eventually intersect $Q$ again. Now the second outside region has been completed, and is clearly disconnected from the first. Therefore $P$ and $Q$ aren't weakly disjoint. 
\end{proof}

\begin{lemma}
\label{lem:one-in-other}
If two convex polygons $P, Q$ are weakly disjoint, but not disjoint, then one contains a vertex of the other.
\end{lemma}

\begin{proof}
If two convex polygons are not disjoint, then they have a non-empty intersection. If this intersection has no area, then they only share part of a boundary. However the vertices are in general position, so this cannot be the case. So their intersection has some non-zero area. Remark that the vertices of $P \cap Q$ are either vertices of $P$, $Q$, or points on $\partial P \cap \partial Q$. Since $P \cap Q$ has positive area, it must have at least $3$ vertices. However, by Lemma~\ref{lem:pseudo-disks}, we know that there are at most two points in $\partial P \cap \partial Q$. So it follows that one of these three vertices must be a vertex of $P$ or $Q$. Therefore a vertex of one is inside the other.
\end{proof}

\begin{lemma}
\label{duality}
The convex hulls of two disjoint trees are weakly disjoint.
\end{lemma}

\begin{proof}
Assume there exists two disjoint trees $R$, $S$, but their convex hulls are not weakly disjoint. Let $P = Conv(R)$ and $Q = Conv(S)$. If $R$ and $S$ share a vertex, then clearly they are not disjoint, and we have a contradiction. Then either $P \setminus Q$ is disconnected, or $Q \setminus P$ is. Assume without loss of generality that $P \setminus Q$ is disconnected. Then there exists two points $p, p' \in P \setminus Q$  such that there exists no path between $p$ and $p'$ inside of $P \setminus Q$. Since both $P$ and $Q$ are convex and share no vertices, the connected components $p$ and $p'$ are part of must contain a vertex of $P$. Therefore,  without loss of generality, we may assume $p$ and $p'$ are vertices of $P$. However, that means $p$ and $p'$ are points on $R$, which has by definition a path that connects them. So either $R$ and $S$ intersect, or there exists a path between $p$ and $p'$; both of which are contradictions. Therefore, if two trees are disjoint, their convex hulls must be weakly disjoint.
\end{proof}

Since the convex hulls of disjoint trees are weakly disjoint, unlike general convex polygons, finding the convex hull of their union is simply a matter of finding at most two tangents to join them by. However, in merging two convex hulls it is no longer guaranteed that the new set of convex hulls has this property. Therefore, it would be desirable to merge convex hulls in some way in which we can maintain this property as an invariant.


\subsection{Shoot and Insert}
If two trees $R$ and $S$ in $T$ have intersecting convex hulls, and we can find an edge to connect $R$ and $S$ without intersecting any other tree in $T$, then we have effectively merged the two trees, while maintaining the invariant of having a set of pairwise non-crossing trees.

\begin{lemma} 
\label{lem:blocked-or-nested}
Assume $R$ and $S$ are two non-crossing trees whose convex hulls intersect. Then the convex hull of one is strictly inside the other, or there exists a pair of adjacent vertices on the convex hull of one whose visibility is blocked by the other tree.  
\end{lemma}

\begin{proof}
By Lemma~\ref{lem:one-in-other}, we know that one contains a vertex of the other. Assume without loss of generality that a vertex $r$ of $Conv(R)$ is inside of $Conv(S)$. If every other vertex of $Conv(R)$ is inside of $Conv(S)$, then $Conv(R)$ is strictly inside of $Conv(S)$ and we are done. Assume this is not the case. Then there exists some path along $R$ from $r$ to the outside of $Conv(S)$. This path must pass between two vertices of $Conv(S)$, and therefore obstruct their visibility.
\end{proof}

Consider shooting a ray between the two vertices $p$, $q$ of $Conv(S)$ that are obstructed by one or more other trees. This ray will necessarily intersect some other tree $R$ first at a point $q'$. By definition, the edge $pq'$ is an edge that joins $R$ and $S$ without intersecting any other tree. If this is the case, then we can stop shooting rays along $S$, replace $S$ and $R$ with $S \cup R \cup pq'$, and starting shooting rays along the convex hull of that new connected component. Furthermore, if we perform this process for all adjacent pairs of vertices of $Conv(S)$, and every ray reached the target vertex, we can conclude that either $S$ is disjoint from all other convex hulls, or part of a well-nested hierarchy of boundary-disjoint convex hulls. If the former, then $S$ is part of our output. If the latter, then the largest convex hull that contains $S$ is part of our output. 

Ishaque et al.\cite{Ishaque12} provide a ray shooting data structure that supports shooting rays from the boundary of obstacles, that are themselves inserted into the obstacles. Using their structure, a set of $n$ pairwise disjoint polygonal obstacles can be preprocessed in $O(n\log n)$ time and space to support $m$ permanent ray shootings in $O((n+m)\log^2n + m\log m)$ time. Therefore shooting $n$ rays takes $O(n\log^2n)$ time. We refer to this data structure as \emph{permashoot}.


\subsection{Algorithm}
We start by computing the sets  
\[ \C = \{ \CH(T_i) | 1 \leq i \leq m \}  
\]
and  
\[ E = \{e | \mbox{ $e$ is an edge of some element of $\C$} \} . 
\] 
We build a permashoot instance $R$ on $T$, and a union-find data 
structure $U$ on $T$. The latter structure is used for determining what 
connected component a given edge is part of.

As long as $E$ is non-empty, we do the following: 
Take an arbitrary edge $e$ in $E$ and remove it from $E$. If 
$e$ is not stored in $R$, search in $U$ for $s$, the component $e$ is 
part of. Shoot a ray in $R$ from one endpoint of $e$ along $e$, and 
return the component $r$ that was hit. If $s \neq r$, then merge 
$\CH(s)$ and $\CH(r)$ in $\C$; remove and add edges from $E$ to 
reflect the new state of $\C$; and union $s$ and $r$ in $U$. 

At this moment, the set $E$ is empty. We perform a plane-sweep on $\C$, 
and return all the convex hulls that are not contained inside another 
convex hull.
  
An example is given in Figure~\ref{fig:algorithm}.

\begin{figure}[ht]
  \centering
  \includegraphics[scale=0.7]{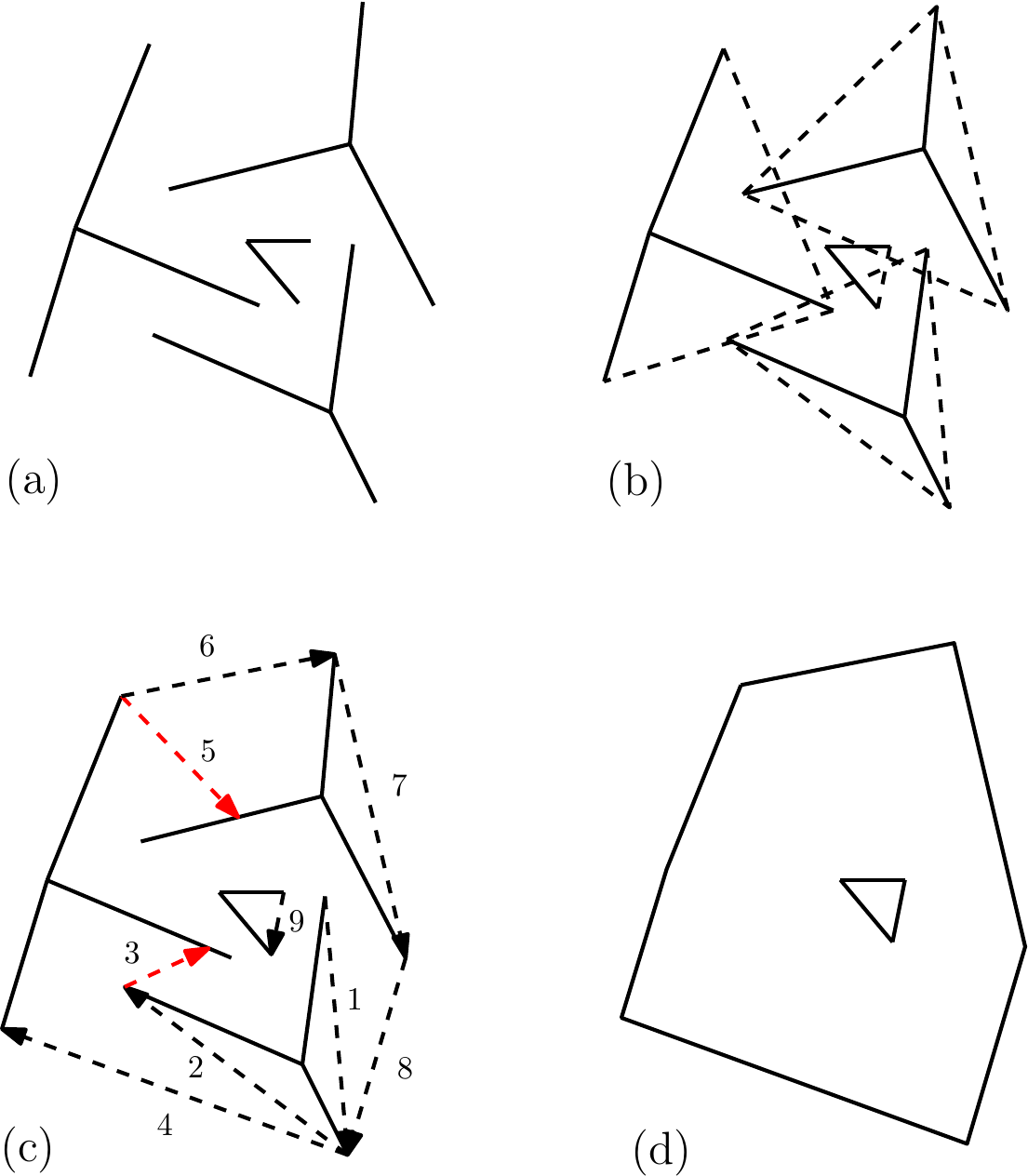}
  \caption{(a)~The input; (b)~Initial convex hulls of the input; (c)~Rays shot by the algorithm (numbered in order they were shot), with rays that caused a merge in red; (d)~Well nested hierarchy of hulls that results}
  \label{fig:algorithm}
\end{figure}

Our algorithm shoots a ray for every edge of every convex hull. If 
any two convex hulls intersect, but are not well-nested, then they 
are found during this process, and replaced by the convex hull of 
their union with an edge that joins them without intersecting any 
other components. This ensures that the invariant of having a set of 
pairwise weakly disjoint polygons holds. This continues until no more 
intersections are found in this way. By 
Lemma~\ref{lem:blocked-or-nested}, we can conclude that we now have 
a set of convex hulls that are either disjoint, or part of a 
well-nested hierarchy. Our plane-sweep then finds all the maximal 
hulls, and returns only these.

\subsection{Analysis}
Because we maintain the invariant of having a set of pairwise 
weakly disjoint polygons, we know that each union adds at most two edges to the set of edges (the tangents between the two hulls). At worst, we perform $O(m) = O(n)$ unions, which adds $O(n)$ edges to check. Initially, there are $O(n)$ edges to check from the starting hulls. Therefore we end up shooting $O(n)$ rays, which takes $O(n\log^2n)$ time. 

For each ray shot we perform a constant amount of union and find 
operations to our union-find structure, each of which can easily be 
done in $O(\log n)$ time \cite[Chapter 21]{clrs}. So union-find only 
takes us $O(n \log n)$ time in total. 

Merging two weakly disjoint convex hulls takes $O(\log n)$ time if we 
maintain them using height balanced binary search 
trees \cite[Section 3.3.7]{smidsbook}. Since we merge at most 
$O(m) = O(n)$ trees, merging the trees takes $O(n\log n)$ time.

Finally, the plane-sweep takes $O(n\log n)$ time to find all the 
maximal convex hulls.

Therefore our algorithm takes $O(n \log^2 n)$ time. This proves the 
following theorem.

\begin{theorem}
The hull-cover of a set of pairwise non-crossing trees with a total 
of $n$ vertices can be computed in $O(n\log^2 n)$ time.
\end{theorem}

\section{Computing the Box-Cover} 
We now assume that $\phi$ is the axis-aligned bounding box cover. 
We refer to the $\phi$-cover as the \emph{box-cover}.  

Let $Box(S)$ be the axis-aligned bounding box of a tree $S$. A 
simple solution to box-cover is as follows. Create a dynamic range 
searching data structure that stores axis-aligned line segments 
and supports queries for those line segments in an axis-aligned 
query box. For each tree $S$ in the input, query the structure for the segments found in the $Box(S)$. For each segment found, remove its parent bounding box from the structure. Then perform a query on the structure with the bounding box of all the boxes found in this way, plus the bounding box we just queried with. Repeat this until no segments are found. Then insert the last box we queried with into the structure. Then run a plane sweep to find all the outermost boxes.

When our algorithm finishes inserting boxes we have a set of boundary-disjoint boxes, as in our hull-cover algorithm. Therefore, as before, it is correct.

Dynamic structures for axis-aligned segment queries exist that take $O(\log^2n + k)$ time for queries, insertion, and deletion\cite{dynamicrange}. Since we start with an empty structure, preprocessing time is irrelevant. When we find an intersection, we replace $O(k)$ boxes with a single box. Since there are $O(m) = O(n)$ boxes, and each box gets inserted and removed at most once, it follows that our algorithm takes $O(n\log^2n)$ time to perform this process in total. The plane sweep takes only $O(n\log n)$ time. Therefore, our algorithm takes $O(n\log^2n)$ time in total. This proves the following theorem.

\begin{theorem}
The box-cover of a set of pairwise non-crossing trees with a total 
of $n$ vertices can be computed in $O(n\log^2 n)$ time.
\end{theorem}

\section{Conclusions and Open Problems}
We are able to compute the solutions to hull-cover and box-cover in $O(n\log^2n)$ time. However this is not obviously optimal. It remains to be seen whether there are better algorithms for these problems.

While the hull-cover is a potentially powerful pre-processing step for computing the actual coverage, the relationship between the two is fairly weak. In the best case the hull-cover is the convex hull of the input, and the two are the same. However in the worst case the hull-cover is exactly the input, but the coverage is something of size $\Omega(n^4)$.

Given a set $\mathcal{O}$ of orientations, an $\mathcal{O}$-convex set $S$ is a set of points such that every line with an orientation in $\mathcal{O}$ has either an empty or connected intersection with $S$. The $\mathcal{O}$-hull of a set $T$ of points is then the intersection of all $\mathcal{O}$-convex sets that contain $T$. When $\mathcal{O} = \{[0, 180)\}$, the $\mathcal{O}$-hull is the convex hull. When $\mathcal{O} = \emptyset$, the $\mathcal{O}$-hull is the identity function.
The $\mathcal{O}$-hull satisfies our properties for being well-defined \cite{rawlins91}. However, an algorithm for the general $\mathcal{O}$-hull is not immediately obvious. Further, it is unclear as to whether there are other non-trivial well-defined covering functions beyond the $\mathcal{O}$-hull and the axis-aligned bounding box. The geodesic hull does satisfy our properties, but without a bounding domain the geodesic hull is just the convex hull. We know from the start of the paper that the minimum enclosing circle does not produce well-defined results, and a similar argument applies to the minimum enclosing square. 

Remark that our proof that general $\phi$-covers are well-defined does not rely on the fact that we are working in two dimensions. This allows us to easily extend the problem into higher dimensions, where the convex hull and bounding box still work. However, while our technique for bounding boxes generalizes to $d$-dimensions nicely, our technique for the convex hull does not. Therefore, a technique for computing the hull-cover that generalizes well would be desirable. 

\section*{Acknowledgement}
Special thanks to Pat Morin for consultation on certain proofs.

\addcontentsline{toc}{chapter}{References}
\small 
\bibliographystyle{abbrv}
\bibliography{bibliography}

\end{document}